\documentclass[11pt,a4paper]{article}
\usepackage[margin=25mm]{geometry}
\usepackage[british]{babel}
\usepackage[T1]{fontenc}
\usepackage{lmodern,microtype,amsmath,amssymb,amsthm,mathtools}
\usepackage{graphicx,booktabs,tabularx,longtable,array,enumitem}
\usepackage[numbers,sort&compress]{natbib}
\usepackage{authblk,hyperref,fancyhdr}
\hypersetup{colorlinks=true,allcolors=black,pdftitle={Equilibrium decisions under partial observation}}
\newtheorem{theorem}{Theorem}[section]
\newtheorem{proposition}[theorem]{Proposition}
\newtheorem{corollary}[theorem]{Corollary}

\theoremstyle{definition}

\theoremstyle{remark}
\newcommand{\R}{\mathbb R}
\newcommand{\spec}{\operatorname{spec}}\newcommand{\diag}{\operatorname{diag}}
\newcommand{\Lip}{\operatorname{Lip}}
\newcommand{\eps}{\varepsilon}\newcommand{\norm}[1]{\lVert#1\rVert}
\newcommand{\calT}{\mathcal T}\newcommand{\calP}{\mathcal P}
\title{Equilibrium decisions under partial observation:\\dynamical error bounds and information requirements}
\author[1]{Ali Ghoroghi\thanks{Corresponding author: ghoroghi@cardiff.ac.uk}}
\author[1]{Yacine Rezgui}\author[1]{Afrouz Ghaemi}
\affil[1]{School of Engineering, Cardiff University, Cardiff, United Kingdom}
\date{}
\begin{document}\maketitle
\begin{abstract}
Reduced monitoring can preserve stable dynamics while changing an equilibrium decision. Motivated by heterogeneous urban-mobility indicator coverage, we study fixed linear observations of a nonlinear recovery flow coupled to a strategic fixed-point map, and ask when a decision based on the reduced model can be trusted. The main result is a joint certificate: bounds on trajectory tracking, equilibrium displacement and unresolved output dependence are assembled into one decision-margin error which, under a stated Gaussian observation model, yields a conservative one-sided test and an explicit information requirement for prescribed worst-case power. In the scalar experiment this requirement is sharp: when the effect does not exceed twice the error budget, no valid test has worst-case power above its size, however many observations are taken. A second result connects manifold-adapted observations to defect bounds through a locally invertible observation chart, while distinguishing observation initialisation from asymptotic phase initialisation. Supporting results, resting on established lumping, contraction and fixed-point theory, cover exact quotients, an observation-fibre obstruction and the dependence of recovery-time comparisons on the tolerance convention. Seven reproducible synthetic experiments examine each component. The results give conditional criteria for certifying partial monitoring and separate mathematical guarantees from empirical validation.
\end{abstract}
\textbf{Subject areas:} applied mathematics; mathematical modelling; complex systems.\\

\section{Introduction}
A monitoring system rarely observes every variable of the system it describes. In urban mobility, routinely recorded throughput and congestion can coexist with sparse information about redundancy, accessibility, behavioural adaptation or communication-system support. A reduced model may nevertheless settle to a stable equilibrium and produce a precise numerical estimate of improvement. The question is whether that estimate supports the same decision as the fuller model, at a specified uncertainty level.

This question has several mathematically distinct parts. Omitted coordinates may change the velocity of observed coordinates, prevent a closed reduced evolution equation, or shift a coupled equilibrium. Even an accurate reduced trajectory may omit variables needed to define the performance score. Finally, measurement noise and observation cadence determine whether an improvement exceeding a chosen margin can be detected. Neither trajectory accuracy nor statistical precision alone settles the combined question.

There are mature theories for the separate parts. Exact lumping characterises autonomous dynamics of aggregated variables \citep{LiRabitz1989,Tomlin1997}; approximate constrained lumping already treats prescribed outputs and controlled reduction errors \citep{Leguizamon2025}. Invariant-manifold and singular-perturbation methods exploit dynamical structure \citep{Carr1981,Fenichel1979,Kuehn2015,Roberts2015}. Contraction theory bounds perturbations of stable flows \citep{LohmillerSlotine1998,Soderlind2006}. Statistical information and experimental design quantify the value of observation channels \citep{Kay1993,Pukelsheim1993,LehmannRomano2005}. Fixed indicator availability is therefore not, by itself, a new reduction problem, and linear observations do not exclude nonlinear manifold parametrisations.

The contribution developed here is a decision-oriented synthesis of these tools for a two-layer equilibrium model. A deterministic error budget combines trajectory error, fixed-point displacement and unresolved dependence of the decision observable on hidden variables. A Gaussian testing result then converts that budget into a conservative decision and a minimum information requirement. Accounting for estimator bias in inference is an established principle \citep{ArmstrongKolesar2020}; here the bias allowance is constructed from the dynamical, strategic and observation errors of a specified reduction. This construction explicitly distinguishes model error from measurement variability: a more precise estimate of a biased reduced quantity need not be a more reliable estimate of the full decision margin.

The paper has two results that we regard as its contribution. The principal one is Theorem~\ref{thm:joint} with Corollary~\ref{cor:minimax}: the assembled error budget, the bias-adjusted test, the information threshold, and the statement that in the scalar Gaussian experiment an effect no larger than twice the budget cannot be detected by any valid test. The second is Proposition~\ref{prop:chart}, which links the geometry of a fixed observation to the dynamical defect entering that budget. Propositions~\ref{prop:exact}, \ref{prop:fibre}, \ref{prop:tracking} and~\ref{prop:fixed} are stated for completeness and for their domain bookkeeping; their content is established, as Section~\ref{sec:inherited} records. The sharpness statement is confined to the scalar experiment and the full bounded-bias class.

Three further distinctions make the synthesis useful. First, exact reproduction of observed trajectories does not imply recovery of hidden equilibria or hidden decision variables. Secondly, failure of a sufficient error bound leaves a reduction uncertified rather than proving it inaccurate. Thirdly, the direction of a recovery-time comparison depends on the threshold convention; spectral inclusion does not justify an unconditional finite-time comparison of separately normalised trajectories.

The motivating application is a fast recovery layer coupled to slower strategic adaptation. The development is local to specified operating domains and fixed strategic conditions; it is not a singular-perturbation theorem for an unspecified continuously coupled traffic system. All numerical evidence is synthetic. The benchmarks examine the mathematical consequences of limited observation, rather than the actual monitoring capacity or performance of a named city. An applied companion, under review at Reliability Engineering and System Safety, addresses disruption assessment; the present paper concerns the conditions under which a reduced mathematical representation can support an equilibrium claim. An earlier and more general version of this theoretical work is available as a preprint \citep{GhoroghiPreprint2026}; the present paper narrows and corrects that treatment.

\section{Motivating model and relation to existing methods}
\subsection{Recovery, performance and strategic equilibrium}
Let $x\in\R^d$ denote normalised indicators and let $z$ denote strategic variables. On an interval between strategic updates, the indicator dynamics are
\begin{equation}\dot x=f(x),\qquad f\in C^1(U,\R^d),\qquad f(x^*)=0.\label{eq:flow}\end{equation}
The operating neighbourhood $U$ and all norms are specified for each result. A mobility interpretation uses twelve coordinates
\begin{equation}
x=(M,R,A,C,D_{\rm eff},S,Q,P,B,E_{\rm energy},E_{\rm ICT},I)^\top.\label{eq:state}
\end{equation}
They represent throughput, redundancy, adaptation, network entropy, disturbance, stress, equity penalty, demand management, compliance, energy support, communication support and cross-layer synergy. These are candidate constructs, not twelve directly measured scalar quantities. Their use is motivated by the AntifragiCity research programme \citep{Rezgui2025}; the numerical parameters below are not estimates from that programme.

A performance score can be an observable $h(x)$, or a separately evolving filtered score
\begin{equation}\dot E=-\kappa\{E-h(x)\},\qquad\kappa>0.\label{eq:filter}\end{equation}
These alternatives are not dynamically equivalent. For fixed $x=x^*$, the second has solution $E(t)-h(x^*)=e^{-\kappa t}\{E(0)-h(x^*)\}$. If $h$ is $L_h$-Lipschitz, the effect of an input-reconstruction error $r(t)$ is at most $e^{-\kappa t}|e_E(0)|+\kappa L_h\int_0^t e^{-\kappa(t-s)}r(s)\,ds$. A discrete Euler update has contraction factor $|1-\kappa\Delta t|<1$ for $0<\kappa\Delta t<2$; at two it is non-decaying, not asymptotically stable.

A possible bounded target is a smooth saturation of a weighted benefit-minus-penalty score,
\begin{equation}
h(x)=s\!\left(\lambda\sum_{j\in\mathcal B}w_j b_j(x_j)
 -(1-\lambda)\sum_{j\in\mathcal D}v_j p_j(x_j)\right),\quad 0<\lambda<1.\label{eq:target}
\end{equation}
Here $s$ is a specified smooth map into $(0,1)$, the weights are nonnegative, and the functions and scales are fixed during the comparison. Rolling renormalisation would change the observable and must be incorporated into the state or model explicitly. Hard clipping, absolute values and piecewise saturation are not $C^1$ at their kinks; the differentiable results below apply on smooth operating regions or to a stated smooth replacement. The theory does not establish that any particular weights measure societal benefit.

On the strategic timescale, write
\begin{equation}u=(x,z),\qquad u=\calT(u)=(F(x,z),G(z,x)).\label{eq:fixed}\end{equation}
A transport implementation might combine user-equilibrium assignment, modal choice and a capacity update. Strictly increasing link costs imply uniqueness of equilibrium \emph{link} flows for fixed feasible demands and capacities through the Beckmann formulation \citep{Beckmann1956}; they do not imply uniqueness of route flows or of a coupled modal-capacity equilibrium. For the abstract analysis, existence is assumed, or obtained when the complete update is a continuous self-map of a nonempty compact convex set by Brouwer's theorem. A discontinuous backtracking step is not automatically continuous. A set-valued formulation needs nonempty compact convex values and upper hemicontinuity for Kakutani's theorem; convexifying an update changes the map and does not prove existence for the unmodified algorithm. Our quantitative results instead impose an explicit contraction certificate on the reduced map.

A decision margin is a scalar $\theta$ whose positive sign indicates an improvement exceeding a declared threshold. For example,
\begin{equation}\theta=R-\tau,\qquad R=\mu_{\rm post}/E_{\rm base},\qquad E_{\rm base}\ge b_0>0.\label{eq:margin}\end{equation}
The baseline, post-event window, score scales and threshold must be fixed independently of the tested noise or modelled jointly. Calling $\theta>0$ evidence of antifragility additionally requires an appropriate causal or counterfactual design. The present paper establishes conditions for preserving this specified improvement decision; it does not identify disruption-induced benefit from an uncontrolled before-after comparison.

\subsection{What is inherited and what is added}\label{sec:inherited}
Exact and approximate lumping are the closest mathematical precedents for the fixed-observation flow. In particular, the invariant-subspace condition and spectral inclusion used below are established reduction facts \citep{LiRabitz1989,Tomlin1997}. Recent approximate constrained lumping develops output-preserving reductions with explicit error control \citep{Leguizamon2025}. We do not claim that an exogenous output or a defect bound alone distinguishes this paper from that literature. The fibre argument of Proposition~\ref{prop:fibre} is the radius-of-information argument of optimal recovery and information-based complexity \citep{TraubWasilkowskiWozniakowski1988}, applied to a decision margin. Certifying the error in a chosen output functional, instead of the whole state, is the aim of goal-oriented error estimation \citep{BeckerRannacher2001}; that literature treats discretisation error of a known model through adjoint weights, whereas the budget here adds strategic displacement and an unobserved-output allowance and passes the total to a hypothesis test.

Manifold reduction addresses a different, overlapping problem: it describes dynamics on an invariant set and the approach of nearby trajectories to it. Physical insight influences both the chosen timescale separation and the parametrisation \citep{Roberts2015}. Suitable measured coordinates may parametrise the manifold even when those coordinates were fixed in advance. Correct initial conditions generally require projection along attracting fibres or asymptotic phase, rather than simply reading the measured coordinate at time zero \citep{Roberts1989}. Section~\ref{sec:geometry} states the observation-chart condition precisely.

Galerkin and balanced reduction supply trajectory and input-output approximations under their own assumptions \citep{HolmesLumleyBerkooz1996,Antoulas2005}. Multilayer-network reductions preserve selected network or process properties \citep{Kivela2014,DeDomenico2015,Boccaletti2014}. Composite-indicator sensitivity analysis addresses weighting and aggregation uncertainty \citep{OECD2008,Saltelli2008}, while identifiability and experimental design address distinguishability of model parameters \citep{WalterPronzato1997,Raue2009,Fedorov1972}. These approaches can be incorporated into the certificate; none is treated as incapable of decision-aware extensions.

Table~\ref{tab:position} locates the paper's construction without attributing novelty to its standard ingredients. Tracking under contraction, fixed-point perturbation, Gaussian information and bias-aware testing are each standard. What we have not found in the literature cited above is their assembly for a coupled recovery and strategic model, namely one budget in which the three deterministic errors appear with explicit gains, together with the consequence that the budget, and not the sample size, decides whether a stated effect is detectable. The main result is the explicit passage from separate approximation guarantees to a robust one-sided decision and its information requirement, including the output ambiguity that approximation of observed coordinates alone cannot remove.
\begin{table}[tb]\centering\small
\caption{Role of established ingredients in the combined decision certificate.}\label{tab:position}
\begin{tabularx}{\textwidth}{p{3.3cm}XX}\toprule
Ingredient & Established role & Role in this paper\\\midrule
Exact/approximate lumping & Closed reduced dynamics and output-error control & A trajectory-error contribution under fixed observations\\
Contraction analysis & Robust tracking and fixed-point perturbation bounds & Explicit gains for dynamic and strategic defects\\
Invariant manifolds & Reduced dynamics and asymptotic initialisation & A chart-conditioned route to bounding observation defects\\
Gaussian information & Estimation precision, power and channel design & Cadence-aware variance in a bias-adjusted decision\\
Joint construction & Uses the ingredients above & A decision-margin budget and robust information threshold\\\bottomrule
\end{tabularx}\end{table}

\section{Observed dynamics and hidden ambiguity}\label{sec:dynamics}
\subsection{Exact quotients and local stability}
An indicator projection is a fixed surjective linear map $\Pi:\R^d\to\R^m$, $m\le d$. Coordinate selectors are the main example. Two states belong to the same observation fibre when $\Pi x_1=\Pi x_2$. Exact projectability on $U$ means that a $C^1$ vector field $g$ satisfies
\begin{equation}\Pi f(x)=g(\Pi x),\qquad x\in U.\label{eq:project}\end{equation}
This is exact linear lumpability. It is equivalent to constancy of $\Pi f$ on each fibre, with the local regularity of $g$ inherited through a linear right inverse. A statement about a quotient on an open neighbourhood is stronger than a statement restricted to a lower-dimensional invariant manifold.

\begin{proposition}[Exact projection and inherited stability]\label{prop:exact}
Suppose $x^*$ is a locally asymptotically stable equilibrium of \eqref{eq:flow} and \eqref{eq:project} holds near it. Then $y^*=\Pi x^*$ is locally asymptotically stable for $\dot y=g(y)$, and projected full trajectories agree with reduced trajectories initialised by $y(0)=\Pi x(0)$ for as long as they remain in the projectability domain. Conversely, this trajectory agreement for every nearby initial state implies \eqref{eq:project}. If $A=Df(x^*)$ and $B=Dg(y^*)$, then $B\Pi=\Pi A$. For a linear flow, projectability is equivalent to $A\ker\Pi\subseteq\ker\Pi$.
\end{proposition}
\begin{proof}
Differentiate $y(t)=\Pi x(t)$ and apply \eqref{eq:project}; uniqueness of the reduced initial-value problem gives agreement. Conversely, differentiation of the assumed agreement at $t=0$ gives \eqref{eq:project}. Choose a linear right inverse $L$ of $\Pi$ and lift any sufficiently small reduced displacement $v$ to $x^*+Lv$. Local asymptotic stability of the full equilibrium keeps sufficiently small lifts inside the projectability neighbourhood, makes their trajectories converge to $x^*$, and bounds their displacements uniformly in terms of the initial displacement. Projection gives attraction and Lyapunov stability of $y^*$ for every sufficiently small $v$. Differentiating \eqref{eq:project} gives $B\Pi=\Pi A$. For linear dynamics, the right side is independent of the representative in a fibre exactly when $\Pi A$ annihilates $\ker\Pi$.
\end{proof}

Hurwitzness of $B$ is sufficient for exponential stability and, in the hyperbolic case, necessary for local asymptotic stability. It is not necessary for nonlinear asymptotic stability: $\dot y=-y^3$ is the elementary counterexample. Under a Hurwitz full Jacobian and exact linear projection, the quotient is automatically Hurwitz by spectral inclusion. These facts should not be combined into an unrestricted ``if and only if'' involving a Hurwitz reduced Jacobian.

Nor does a quotient have to distinguish every full equilibrium. For example, $\dot x=-x$, $\dot z=z-z^3$, observed through $\Pi(x,z)=x$, has an exact stable quotient while the two attracting equilibria $(0,1)$ and $(0,-1)$ are indistinguishable. This is acceptable for predicting $x$ but unacceptable for deciding whether $z$ is positive. The decision observable needs its own condition.

\begin{proposition}[Unresolved output dependence]\label{prop:fibre}
Let $h$ be a real decision margin on a state domain $D$. If $h$ takes both positive and negative values on one observation fibre, no deterministic classifier using only that observation can reproduce the sign of $h$ for every state on that fibre. If $h$ is bounded on each fibre, the best possible uniform scalar approximation error over unrestricted fibrewise representatives is
\begin{equation}
\inf_{\tilde h}\sup_{x\in D}|h(x)-\tilde h(\Pi x)|
 =\frac12\sup_{y\in\Pi D}\left\{\sup_{\Pi x=y}h(x)-\inf_{\Pi x=y}h(x)\right\}.\label{eq:fibre_radius}
\end{equation}
No continuity of the minimising representative is asserted.
\end{proposition}
\begin{proof}
A deterministic label is constant on a fibre and cannot equal two different signs. For scalar approximation, any number approximating two fibre values differs from one by at least half their separation. Taking suprema proves the lower bound. The midpoint of the fibre supremum and infimum attains that radius fibrewise, proving equality.
\end{proof}

The analogous velocity-fibre diameter gives a lower bound on the best closure defect: $\inf_g\sup_D\norm{\Pi f-g\circ\Pi}\ge\frac12\sup_{\Pi x_1=\Pi x_2}\norm{\Pi f(x_1)-\Pi f(x_2)}$. The scalar equality need not extend to vector-valued closures. These are information obstructions for instantaneous deterministic observation, of the kind studied in information-based complexity \citep{TraubWasilkowskiWozniakowski1988}. Delay coordinates, interventions, prior restrictions or additional sensors may change the available information and the admissible domain.

\subsection{Approximate tracking and non-normal transients}
For a candidate closure $g$, define a domain-specific defect
\begin{equation}\eps_\Pi=\sup_{x\in U}\norm{\Pi f(x)-g(\Pi x)}.\label{eq:defect}\end{equation}
Let $x(t)\in U$, $y(t)\in V$ and $\Pi x(t)\in V$ for $t\in[0,T]$, where $V$ is convex. Write $e(t)=\Pi x(t)-y(t)$ and $d(t)=\Pi f(x(t))-g(\Pi x(t))$. The averaged Jacobian
\begin{equation}J_e(t)=\int_0^1Dg(y(t)+se(t))\,ds\label{eq:averaged}\end{equation}
satisfies $\dot e=J_e(t)e+d(t)$ exactly. In several dimensions, a single intermediate point generally cannot replace this matrix integral.

\begin{proposition}[Tracking error bounds]\label{prop:tracking}
If the logarithmic norm satisfies $\mu(Dg(v))\le-\alpha<0$ throughout $V$, then
\begin{equation}
\norm{e(t)}\le e^{-\alpha t}\norm{e(0)}+
 \frac{\eps_\Pi}{\alpha}(1-e^{-\alpha t})=:D_\Pi(t).\label{eq:contract_bound}
\end{equation}
More generally, if the transition operator of \eqref{eq:averaged} satisfies $\norm{\Phi_e(t,s)}\le M e^{-\alpha(t-s)}$, $M\ge1$, then
\begin{equation}
\norm{e(t)}\le M e^{-\alpha t}\norm{e(0)}+
 \frac{M\eps_\Pi}{\alpha}(1-e^{-\alpha t}).\label{eq:transient}
\end{equation}
Both conclusions hold on the stated interval. A limiting bound requires the same assumptions for all $t\ge0$.
\end{proposition}
\begin{proof}
Convexity of the logarithmic norm and \eqref{eq:averaged} give $D^+\norm e\le-\alpha\norm e+\eps_\Pi$. Integrating proves \eqref{eq:contract_bound}. In the second case, variation of constants gives $e(t)=\Phi_e(t,0)e(0)+\int_0^t\Phi_e(t,s)d(s)\,ds$. Applying the two uniform bounds proves \eqref{eq:transient}.
\end{proof}

For a stable but non-normal constant Jacobian $B$, a computable certificate is $P\succ0$ satisfying $B^\top P+PB\preceq-2\alpha_0P$. It gives $M=\sqrt{\lambda_{\max}(P)/\lambda_{\min}(P)}$ in the Euclidean norm. For a nonlinear closure with $\sup_V\norm{Dg-B}_2\le r$, variation of constants and Gronwall's inequality give the transition bound with $\alpha=\alpha_0-Mr>0$. This is a sufficient certificate and may be conservative. Stability of each instantaneous Jacobian alone does not imply the transition-operator bound. If only $\Lip(g)\le L$ is known, the finite-time estimate $\norm e(t)\le e^{Lt}\norm e(0)+\eps_\Pi(e^{Lt}-1)/L$ remains available, with its continuous limit at $L=0$.

The supremum in \eqref{eq:defect} is a model-domain quantity. The largest residual on a finite calibration sample is not a certified upper bound on it. For example, if the residual map is $L_d$-Lipschitz and calibration points form an $h$-net of $U$, then $\eps_\Pi\le\max_i\norm{d(x_i)}+L_dh$; an additional model-error allowance is needed if $f$ is itself estimated. Without coverage or uncertainty control, residuals are diagnostics rather than uniform certificates.

We use three outcomes at a stated tolerance: \emph{certified} when the assumptions and bound establish it, \emph{demonstrably failing} when an actual counterexample exceeds it, and \emph{uncertified} otherwise. A large upper bound, or a nonnegative Euclidean logarithmic norm, is not proof of inaccuracy.

\section{Coupled equilibrium under reduction}\label{sec:fixed}
Let $\calP=\diag(\Pi_x,\Pi_z)$ reduce the two-layer state in \eqref{eq:fixed}, and let $\widetilde{\calT}$ be a reduced update on a nonempty closed set $D_r$. Assume $\calP u^*\in D_r$. In component norms, suppose its two components have nonnegative Lipschitz comparison matrix
\begin{equation}K=\begin{pmatrix}a&b\\d&c\end{pmatrix}.\label{eq:comparison}\end{equation}
For $\nu\gg0$, use $\norm{(y,w)}_\nu=\max\{\norm y/\nu_1,\norm w/\nu_2\}$. If $K\nu\le q\nu$ with $q<1$, the reduced map is a $q$-contraction in this norm. When $\rho(K)<1$, such a vector exists for any chosen $q\in(\rho(K),1)$: take $\nu=(I-K/q)^{-1}{\bf1}$. The exact choice $q=\rho(K)$ with a positive vector is guaranteed for irreducible $K$; it is not automatic for a reducible matrix \citep{BermanPlemmons1994}.

\begin{proposition}[Equilibrium displacement and solver error]\label{prop:fixed}
Suppose $\widetilde{\calT}:D_r\to D_r$ is a $q$-contraction, $q<1$, and $u^*=\calT(u^*)$. If
\begin{equation}\norm{\widetilde{\calT}(\calP u^*)-\calP\calT(u^*)}_\nu\le\delta_c,\label{eq:commute}\end{equation}
then $\widetilde{\calT}$ has a unique fixed point $v^*$ in $D_r$ and
\begin{equation}\norm{v^*-\calP u^*}_\nu\le\frac{\delta_c}{1-q}.\label{eq:displacement}\end{equation}
For its iteration, $\norm{v_k-v^*}_\nu\le q^k\norm{v_0-v^*}_\nu$. For any approximate output $\hat v\in D_r$, the residual gives
\begin{equation}\norm{\hat v-v^*}_\nu\le\frac{\norm{\hat v-\widetilde{\calT}(\hat v)}_\nu}{1-q}.\label{eq:solver}\end{equation}
\end{proposition}
\begin{proof}
The closed domain is complete and the self-map is a contraction, so Banach's theorem gives existence, uniqueness and geometric convergence \citep{Deimling1985}. Insert $\widetilde{\calT}(\calP u^*)$ between $v^*=\widetilde{\calT}(v^*)$ and $\calP u^*=\calP\calT(u^*)$. The triangle inequality yields $\norm{v^*-\calP u^*}_\nu\le q\norm{v^*-\calP u^*}_\nu+\delta_c$; rearrangement proves \eqref{eq:displacement}. The same argument with $\hat v$ proves \eqref{eq:solver}.
\end{proof}

Exact commutation on a domain is sufficient for $\delta_c=0$; only a residual at the full fixed point is needed for the perturbation estimate. A uniform commutation-defect bound is useful when that fixed point is uncertain. The certified iteration count for relative error $\epsilon$ is $\lceil\log(1/\epsilon)/[-\log q]\rceil$. It grows as $(1-q)^{-1}$ for fixed $\epsilon$ as $q\uparrow1$, but an upper bound growing does not prove that actual solver cost must grow. An increment-based stopping rule must be converted by \eqref{eq:solver} to a fixed-point error bound.

The boundary $q=1$ marks loss of this certificate. It is not a general bifurcation theorem. A loose comparison matrix may reach one while the true map remains strongly contractive. Even a map with derivative one at a fixed point can have an isolated attracting fixed point, as $v\mapsto v-v^3$ does locally. Branches or continua require additional nonlinear analysis. The experiments below therefore distinguish a measured change of behaviour from failure of the sufficient bound.

\section{Information and a combined decision certificate}\label{sec:decision}
\subsection{Observation model, cadence and power}\label{sec:obsmodel}
Let $S$ index a chosen set of observation channels, and stack all observations in a window, including their actual sampling times. The specified local statistical model is
\begin{equation}V_S\sim N(\mu_{0,S}+\vartheta a_S,\Omega_S),\qquad\Omega_S\succ0,\label{eq:gauss}\end{equation}
where $\mu_{0,S}$, $a_S$ and $\Omega_S$ are known and $\Omega_S$ does not depend on $\vartheta$. In the simplest state model $a_S=C_Su_\vartheta$; a time-window model repeats or varies these loadings and includes temporal covariance in $\Omega_S$. Random latent-state fluctuations must be included in that covariance. In general, rules assembling constructs from measured quantities do not automatically provide an inverse observation map $C_S$.

The window Fisher information and efficient estimator are
\begin{equation}
J_S=a_S^\top\Omega_S^{-1}a_S,\qquad
\hat\vartheta_S=J_S^{-1}a_S^\top\Omega_S^{-1}(V_S-\mu_{0,S}),\quad J_S>0.\label{eq:info}
\end{equation}
Thus $\hat\vartheta_S\sim N(\vartheta,J_S^{-1})$ exactly under \eqref{eq:gauss}. For $N$ independent repeated observations with per-observation information $\mathcal I_S$, $J_S=N\mathcal I_S$. That replacement is not valid for arbitrary dependent observations. A known-covariance Gaussian calculation is an exact benchmark, not a validation of an estimated-covariance procedure.

For a nested addition of a scalar observation with sensitivity $a_j$, variance $\omega_j$ and covariance $\Omega_{Sj}$, a block inverse gives
\begin{equation}
J_{S\cup j}-J_S=
\frac{(a_j-\Omega_{jS}\Omega_S^{-1}a_S)^2}
{\omega_j-\Omega_{jS}\Omega_S^{-1}\Omega_{Sj}}\ge0.\label{eq:innovation}
\end{equation}
This proves information monotonicity for genuinely nested observations under one joint model. It does not compare arbitrary non-nested indicator sets. A channel with zero marginal sensitivity can be useful as a correlated noise reference.

With zero model bias, the level-$\alpha$ test rejects when $\hat\vartheta_S>z_{1-\alpha}/\sqrt{J_S}$, with power
\begin{equation}\pi_S(\vartheta)=1-\Phi(z_{1-\alpha}-\vartheta\sqrt{J_S}).\label{eq:power}\end{equation}
Here $z_p$ is the $p$-quantile of the standard normal distribution. For an oracle containing $S$, the nonnegative power loss is $\pi_{\rm all}-\pi_S=\Phi(z_{1-\alpha}-\vartheta\sqrt{J_S})-\Phi(z_{1-\alpha}-\vartheta\sqrt{J_{\rm all}})$ for $\vartheta>0$. Independent additional calibration noise of variance $\sigma_{\rm cal}^2$ changes the estimator variance to $v_S=J_S^{-1}+\sigma_{\rm cal}^2$. The power formulae below additionally require this noise to be Gaussian; independence and a finite variance alone do not justify them. Such a term must be separately justified; information already lost by omission must not be counted a second time as an unexplained residual variance.

For zero bias and effect $\vartheta^*>0$, target power $1-\beta>\alpha$ requires
\begin{equation}
J_S\ge\left[\left(\frac{\vartheta^*}{z_{1-\alpha}+z_{1-\beta}}\right)^2-\sigma_{\rm cal}^2\right]^{-1},\label{eq:plain_threshold}
\end{equation}
provided the bracket is positive. With $J_S=N\mathcal I_S$, divide the right side by $N$ to obtain the per-observation requirement; there is no extra factor of $N$ multiplying $\sigma_{\rm cal}^2$ inside the bracket.

\subsection{From model defects to a trustworthy decision}
Four scalar quantities must be kept apart in this subsection. The \emph{full decision margin} $\theta$ is the quantity whose sign is wanted. The \emph{reduced decision margin} $\tilde\theta$ is what the reduced model computes; it differs from $\theta$ by a deterministic model bias. The \emph{estimator} $\hat\theta_S$ is what the data deliver; it differs from $\tilde\theta$ by measurement noise. The parameter $\vartheta$ of Section~\ref{sec:obsmodel} is a generic Gaussian mean. Theorem~\ref{thm:joint} bounds the first difference and \emph{assumes} a Gaussian law for the second. That assumption holds when $\tilde\theta$ is the mean parameter $\vartheta$ of the linear model \eqref{eq:gauss} and $\hat\theta_S$ is the estimator \eqref{eq:info}. For a margin that is a nonlinear function of the observed channels, such as the ratio \eqref{eq:margin}, the Gaussian law of $\hat\theta_S$ is an additional approximation that the theorem does not supply.

Throughout, $x(t)$ is the recovery-layer trajectory at the decision time, $u^*$ is the strategic fixed point of \eqref{eq:fixed}, and $H$ denotes the reduced scalar output map; it is unrelated to the strategic update component $G$ in \eqref{eq:fixed}.

Suppose the full decision margin at a specified time $t$ and strategic equilibrium $u^*$ can be represented by
\begin{equation}
\theta=h(x(t),u^*)=H(\Pi x(t),\calP u^*)+r_h,
\qquad |r_h|\le\eta_h.\label{eq:output}
\end{equation}
The quantity $\eta_h$ is an explicit output-reconstruction allowance, not an assumption that hidden indicators are irrelevant. Proposition~\ref{prop:fibre} supplies a lower bound on how small it can be. Assume $H$ is separately Lipschitz with constants $L_x,L_u$, in the norms of the tracking and fixed-point bounds, on a product domain containing the observed states $\Pi x(t),y(t)$ and the fixed-point states $\calP u^*,v^*$ at every pair used in the proof. Define the reduced decision margin $\tilde\theta=H(y(t),v^*)$.

\begin{theorem}[Joint bias and information certificate]\label{thm:joint}
Under \eqref{eq:output}, Proposition~\ref{prop:tracking} with error bound $D_\Pi(t)$, and Proposition~\ref{prop:fixed},
\begin{equation}
|\tilde\theta-\theta|\le B_S(t):=\eta_h+L_xD_\Pi(t)+L_u\frac{\delta_c}{1-q}.\label{eq:bias}
\end{equation}
Suppose, conditional on the fixed model and independently selected support, an estimator satisfies $\hat\theta_S=\tilde\theta+\xi$ with $\xi\sim N(0,v_S)$ and known $v_S=J_S^{-1}+\sigma_{\rm cal}^2$. Write $B=B_S(t)$. Then the conservative rule
\begin{equation}\hat\theta_S>B+z_{1-\alpha}\sqrt{v_S}\label{eq:robust_test}\end{equation}
has rejection probability at most $\alpha$ for every $\theta\le0$ and every bias in $[-B,B]$. At $\theta=\theta^*>0$, its worst-case power over that bias interval is
\begin{equation}
\pi_{\rm robust,min}=1-\Phi\!\left(z_{1-\alpha}-\frac{\theta^*-2B}{\sqrt{v_S}}\right).\label{eq:robust_power}
\end{equation}
For $1-\beta>\alpha$, attaining that worst-case power requires $\theta^*>2B$ and
\begin{equation}
J_S\ge\left[\left(\frac{\theta^*-2B}{z_{1-\alpha}+z_{1-\beta}}\right)^2-\sigma_{\rm cal}^2\right]^{-1},\label{eq:joint_threshold}
\end{equation}
with a positive bracket. For this rule and the stated Gaussian uncertainty class, these conditions are also sufficient.
\end{theorem}
\begin{proof}
Add and subtract $H(\Pi x(t),v^*)$ in $\tilde\theta-h(x(t),u^*)$, apply the two Lipschitz bounds and $|r_h|\le\eta_h$, then use the tracking and equilibrium bounds to obtain \eqref{eq:bias}. Write $b=\tilde\theta-\theta\in[-B,B]$. Under the null, $\theta+b\le B$, so \eqref{eq:robust_test} requires a Gaussian error exceeding at least $z_{1-\alpha}\sqrt{v_S}$. Its probability is at most $\alpha$. Under the alternative, rejection probability increases with $b$, so its minimum occurs at $b=-B$, yielding \eqref{eq:robust_power}. Solving $\pi_{\rm robust,min}\ge1-\beta$ gives $\theta^*-2B\ge(z_{1-\alpha}+z_{1-\beta})\sqrt{v_S}$. The sum of quantiles is positive when $1-\beta>\alpha$. Substitution of $v_S$ and rearrangement give the stated conditions.
\end{proof}

\begin{corollary}[Sharp information obstruction in the scalar experiment]\label{cor:minimax}
In the scalar Gaussian experiment of Theorem~\ref{thm:joint}, with $B$ and $v_S$ fixed, if $0<\theta^*\le2B$, no test that has size at most $\alpha$ uniformly over $\theta\le0$ and $b\in[-B,B]$ can have worst-case power exceeding $\alpha$ at $\theta^*$. If $\theta^*>2B$, rule~\eqref{eq:robust_test} maximises worst-case power among all such tests, including randomised tests.
\end{corollary}
\begin{proof}
For $\theta^*\le2B$, the null choice $(\theta,b)=(0,\theta^*/2)$ and alternative choice $(\theta^*,-\theta^*/2)$ give the same normal distribution. Any test has identical rejection probability at those choices. For $\theta^*>2B$, every uniformly valid test must have size at most $\alpha$ for the simple null mean $B$. Its power at the allowed alternative mean $\theta^*-B>B$ is bounded by the Neyman--Pearson upper-tail test \citep{LehmannRomano2005}. That test is exactly \eqref{eq:robust_test}, and its power increases with the mean, so it attains this bound at the worst alternative bias $-B$ while controlling all null means.
\end{proof}

This optimality concerns the stipulated scalar experiment and the full bounded-bias class. The class can be larger than the biases attainable by a particular physical model, so the conclusion does not make the additive model-error bound sharp.

The factor $2B$ has a specific origin: one allowance protects the null against upward bias, and another represents downward bias under the least favourable alternative. It is not an assertion that the actual bias always equals this worst case. If its sign or value is known, sharper tests are possible. If $\theta^*\le2B$, more observations alone cannot give this rule worst-case power exceeding $\alpha$. Even for $\theta^*>2B$, calibration uncertainty can set an additional floor.

A finite solver residual adds $L_u\norm{\hat v-\widetilde{\calT}(\hat v)}_\nu/(1-q)$ to $B$. For a time-averaged decision, replace $D_\Pi(t)$ by the corresponding average error bound. For a filtered performance score, use the convolution following \eqref{eq:filter}. Ratio decisions such as \eqref{eq:margin} require a positive baseline bound: if both baseline values exceed $b_0$ and $|\tilde\mu|\le M_\mu$, then $|\mu/b-\tilde\mu/\tilde b|\le|\mu-\tilde\mu|/b_0+M_\mu|b-\tilde b|/b_0^2$. Threshold uncertainty must also enter the margin error.

The theorem is conditional on honest deterministic bounds and the stipulated Gaussian observation model. A bound estimated with coverage $1-\gamma$ yields at best a union-bound size guarantee $\alpha+\gamma$ unless its failure probability is allocated within the target level. Selecting channels or tuning bounds with the same data used for the test requires separate calibration, sample splitting or selection-aware inference. These conditions are part of the certificate, rather than details repaired by a large sample size.

\subsection{Selecting an indicator support}
For independent scalar channels, $J_S=\sum_{j\in S}a_j^2/\sigma_j^2$. Sorting by these contributions maximises information for every cardinality and finds the smallest statistical support exceeding a \emph{fixed} information threshold. Independent blocks have the same property when blocks are equal-cost indivisible items and cardinality counts blocks. Unequal costs, within-block selection, or support-dependent $B_S$ invalidate that simple optimality claim.

With correlated observations, information is monotone but need not be submodular. An explicit example has
\begin{equation}
a=(1,0.9,0)^\top,\quad
\Omega=\begin{pmatrix}1&0&0\\0&1&0.99\\0&0.99&1\end{pmatrix}.\label{eq:correlated_example}
\end{equation}
Greedy information selection at budget two chooses channels $1,2$ and obtains $1.81$, whereas channels $2,3$ give $0.81/(1-0.99^2)\simeq40.70$. The third channel has no marginal signal but removes much of the second channel's noise. Thus correlation is not synonymous with redundancy. The example is illustrative; failure of greedy selection under correlated noise is not claimed as a new phenomenon.

A joint design problem is
\begin{equation}\min_{S}\operatorname{cost}(S)\quad\text{subject to the domain, error and power conditions of Theorem~\ref{thm:joint}.}\label{eq:design}\end{equation}
There is no general greedy solution. More channels increase $J_S$ for a nested statistical model, but a newly chosen dynamical closure may change $D_\Pi$, $q$ and $\eta_h$ non-monotonically. Exhaustive search is feasible for small catalogues; otherwise a proposed support should be reported as certified feasible unless a separate optimisation guarantee is available. The two-channel statistical result in the experiments is not a proof that two channels preserve the full dynamic or strategic decision.

\section{Manifold-adapted observations and initial conditions}\label{sec:geometry}
This section answers three separate questions in turn: when a fixed observation parametrises an invariant manifold (the observation chart); how far the resulting closure is in error off the manifold (the chart-to-defect bound); and which reduced initial condition should be used (observation matching against asymptotic phase).

Suppose a $k$-dimensional invariant manifold is parametrised by $v:A\subset\R^k\to\R^d$ with dynamics $\dot a=\phi(a)$, so $f(v(a))=Dv(a)\phi(a)$. For $m=k$, the observed map is $\psi(a)=\Pi v(a)$. Its derivative is $D\psi=\Pi Dv$. If $\det(\Pi Dv(a_0))\ne0$, the inverse-function theorem supplies a local observation chart. The reduced field on that chart is
\begin{equation}
g(y)=\Pi f(v(\psi^{-1}(y))).\label{eq:chartflow}
\end{equation}
It is exact on the manifold. The correct geometric condition is that the observation kernel intersects the tangent space trivially: $\ker\Pi\cap T_{v(a)}\mathcal M=\{0\}$. Describing this in terms of an unspecified ``data axis'' being tangent or orthogonal is ambiguous. With more observations than manifold dimensions, full column rank gives an immersion locally, but the image is not generally an open reduced state space and needs a separate constrained description.

\begin{proposition}[A local chart-to-defect bound]\label{prop:chart}
Assume a single injective chart $\psi:A\to W$ as above, with convex $W$ and $\sup_W\norm{D\psi^{-1}}\le\beta^{-1}$. Suppose $v$ is $L_v$-Lipschitz on $A$, $f$ is $L_f$-Lipschitz on the relevant neighbourhood, and for every $x$ in a tube $U_r$ there is $a\in A$ with $\norm{x-v(a)}\le r$, $\Pi x\in W$. Then the closure \eqref{eq:chartflow} satisfies
\begin{equation}
\eps_\Pi\le\norm\Pi L_f\left(1+\frac{L_v\norm\Pi}{\beta}\right)r.\label{eq:chartdefect}
\end{equation}
\end{proposition}
\begin{proof}
Set $\hat a=\psi^{-1}(\Pi x)$. Convexity of $W$ and the derivative bound give $\norm{\hat a-a}\le\beta^{-1}\norm{\Pi x-\Pi v(a)}\le\beta^{-1}\norm\Pi r$. Consequently $\norm{x-v(\hat a)}\le r+L_v\norm{\hat a-a}\le(1+L_v\norm\Pi/\beta)r$. Substituting into $\norm{\Pi f(x)-\Pi f(v(\hat a))}\le\norm\Pi L_f\norm{x-v(\hat a)}$ proves the result.
\end{proof}

This is a sufficient local link between observation geometry and the dynamical defect. It is not an identity between conditioning and projectability. At $r=0$ the defect is zero on any valid chart regardless of its conditioning. A small $\beta$ can enlarge this bound without forcing the actual defect to grow. Pointwise nonsingularity alone is not global injectivity, and a pointwise singular-value bound needs the stated inverse-chart domain to produce the global Lipschitz estimate. Parameter units also matter; singular values should be interpreted under declared coordinate scales, or computed with an orthonormal tangent basis for geometric comparisons.

The parabola example supplied in correspondence by A.~J.~Roberts makes the distinction concrete:
\begin{equation}\dot x=-xy,\qquad\dot y=-y+x^2-2y^2,\qquad\mathcal M=\{(a,a^2)\}.\label{eq:roberts}
\end{equation}
On the manifold $\dot a=-a^3$. Observing $s=x+y$ gives $\psi(a)=a+a^2$ and, on the branch $a>-1/2$,
\begin{equation}
a(s)=\frac{\sqrt{1+4s}-1}{2},\quad s>-1/4,\qquad
\dot s=-(1+2a(s))a(s)^3.\label{eq:parabola}
\end{equation}
The inverse derivative is $1/(1+2a)$, which diverges as $a\downarrow-1/2$. Thus a fixed measured combination can parametrise a nonlinear manifold exactly, but a branch restriction and conditioning matter. This example is an illustration of parametrisation, not a new centre-manifold theorem.

Observation and asymptotic initialisation must also be distinguished. The choice $y(0)=\Pi x(0)$ matches the immediate observation; an off-manifold trajectory need not asymptotically track the on-manifold trajectory launched from that value. Roberts's projection construction \citep{Roberts1989,Roberts2015} instead accounts for the attracting fibres. Locally, projection covectors collected in $Z(a)$ determine the phase through an orthogonality condition of the form $Z(a)^\top\{x_0-v(a)\}=0$. This linear relation describes the tangent plane, at $v(a)$, of the attracting fibre through that point. It is therefore exact only to leading order: it recovers the asymptotic phase with an error that is in general quadratic in the distance of $x_0$ from $\mathcal M$ \citep{Roberts1989,Roberts2015}, and it is exact when the attracting fibres are affine. The curvature that matters here is that of the fibres, which is a separate property from the curvature of $\mathcal M$ itself. It uses dynamical information, and is not generally a function of the retained instantaneous observation alone.

A proposed diagnostic based on degeneracy of such projection covectors and the tangent space is distinct from singularity of $\Pi Dv$. We use the latter, explicitly stated chart condition; we do not claim to have proved a general equivalence with Roberts's suggested diagnostic. Manifold identification can also be data driven, so manifold methods do not universally presuppose a fully known vector field. Any data-driven use of \eqref{eq:chartdefect}, however, must account for uncertainty in the chart, dynamics and tube radius.

\section{What reduced monitoring can say about recovery time}\label{sec:recovery}
For a stable linear system $\dot x=Ax$ with exact projection $\Pi A=B\Pi$, a basis adapted to $\ker\Pi$ makes $A$ block triangular. Therefore
\begin{equation}
\spec(A)=\spec(A|_{\ker\Pi})\cup\spec(B),\qquad s(B)\le s(A)<0,\label{eq:spectra}
\end{equation}
where $s(A)=\max\Re\spec(A)$. This spectral inclusion is a standard lumping property \citep{Tomlin1997}. Recovery-time conclusions depend additionally on the norm and threshold.

\begin{proposition}[Common-threshold and relative-threshold recovery]\label{prop:recovery}
Suppose $\norm{\Pi x}\le\norm x$. With a common absolute threshold $\epsilon>0$, define
\begin{equation}T_x(\epsilon)=\inf\{t\ge0:\norm{x(s)}\le\epsilon\text{ for all }s\ge t\}\label{eq:absrec}\end{equation}
and define $T_{\Pi x}(\epsilon)$ identically for the projected trajectory. Then $T_{\Pi x}(\epsilon)\le T_x(\epsilon)$. This fact requires no projectability assumption. In contrast, if each trajectory uses its own relative threshold $\delta\norm{x(0)}$ and $\delta\norm{\Pi x(0)}$, there is no universal finite-time ordering, even for diagonal Hurwitz $A$ and an exact coordinate projection.
\end{proposition}
\begin{proof}
Whenever the full trajectory remains inside the absolute tolerance, non-expansiveness keeps its projection there. For the second statement take $A=\diag(-1,-10)$, $x(0)=(1,100)^\top$ and $\Pi=(1,0)$. At $\delta=0.1$, the projected relative restoration time is $\log 10\simeq2.302585$. The full restoration time solves $e^{-2t}+10^4e^{-20t}=0.01(10001)$ and is approximately $0.230570$. The reduced relative time is larger. Conversely, retaining a fast mode while exciting a hidden slow mode gives the opposite ordering.
\end{proof}

For a general bounded projection, use threshold $\norm\Pi\epsilon$ in the reduced norm to obtain the corresponding common-scale result. If the projected initial displacement is zero, a normalised projected recovery fraction is undefined; recording zero recovery time is a no-visible-displacement convention, not evidence that the hidden state recovered immediately.

\begin{proposition}[Asymptotic envelope ordering]\label{prop:asym}
For Hurwitz $A,B$ in \eqref{eq:spectra}, define operator-norm envelopes $\mathcal T_A(\delta)=\inf\{t:\sup_{s\ge t}\norm{e^{As}}\le\delta\}$ and analogously for $B$. Then
\begin{equation}\lim_{\delta\downarrow0}\frac{\mathcal T_B(\delta)}{\mathcal T_A(\delta)}=\frac{s(A)}{s(B)}\le1.\label{eq:asym}
\end{equation}
For a fixed initial state with nonzero projection, the same limit for separately normalised trajectory times holds with the largest real parts of the excited generalised eigenspaces in place of $s(A),s(B)$.
\end{proposition}
\begin{proof}
Jordan normal form and norm equivalence give upper bounds of the form $C(1+t^r)e^{s(A)t}$, and the spectral radius gives $\norm{e^{At}}\ge e^{s(A)t}$ for an induced norm. Thus $\mathcal T_A(\delta)=\log(1/\delta)/[-s(A)]+O(\log\log(1/\delta))$. The analogous relation for $B$ and spectral inclusion prove \eqref{eq:asym}. For a fixed vector, restrict to its excited generalised eigenspaces. In Jordan coordinates the leading nonzero terms have a fixed exponential rate times polynomial factors; invertibility of the coordinate transform preserves that rate. The relation $e^{Bt}\Pi v=\Pi e^{At}v$ ensures that the projected leading rate cannot exceed the full one.
\end{proof}

Equality of the limiting rates does not imply equality of finite-time horizons. Non-normal transients and different normalising amplitudes can remain consequential at practical tolerances. A claim about recovery commitments should therefore state whether it concerns a common physical tolerance, a relative trajectory fraction or a small-tolerance asymptotic rate.

\section{Reproducible synthetic experiments}\label{sec:experiments}
All experiments were executed with the supplied Python implementation, fixed seeds derived from 20261002, and NumPy, SciPy and Matplotlib. Equations, parameters, output CSVs and verification assertions are supplied together. The archived version-8 code was first rerun: all five CSVs and the parsed scalar summary reproduced exactly. That establishes reproducibility of that archive, not correctness of its mathematical interpretation. The following revised experiments replace the numerical claims in the main text.

\subsection{E1: domain-valid tracking certificates}
The twelve-dimensional linear recovery matrix has diagonal rates $-0.4,-0.5,-0.3$ for $(M,S,D_{\rm eff})$, $-0.05$ for $(R,Q)$ and $-0.15$ for the remaining coordinates. For each $j\in\{R,A,Q,I,E_{\rm ICT}\}$, set $A_{Mj}=c$, $A_{Sj}=-c$, and set $A_{RM}=c/2$. All other off-diagonal entries vanish. We use $c\in\{0,0.02,\ldots,0.18\}$. At $c=0.2$ the $(M,R)$ block has a zero eigenvalue, so that endpoint is excluded from the stable experiment.

The nested supports are $S_1=\{M,S,D_{\rm eff}\}$, $S_2=\{M,R,C,S,Q,I,D_{\rm eff}\}$ and all twelve coordinates. The standard support contains seven coordinates: retaining $D_{\rm eff}$ makes the nesting literal. These names denote synthetic observation designs, not city assignments. For a selector $\Pi$, take $B=\Pi A\Pi^\top$. On the radius-$0.35$ Euclidean ball, the defect is exactly $0.35\norm{\Pi A-B\Pi}_2$.

For each matrix, solve $A^\top P_A+P_AA=-I$ and draw 60 initial conditions on the ellipsoid $x^\top P_Ax=0.35^2\lambda_{\min}(P_A)$. This ellipsoid is invariant and contained in the ball used for the defect. Full and reduced trajectories use matrix exponentials on a grid of 1201 times over $[0,300]$, rather than Euler trajectories tested against a continuous-time theorem. Both the entire sampled horizon and the domain condition are checked.

The Euclidean logarithmic-norm certificate is compared with a Lyapunov certificate obtained by taking $\alpha_0=-s(B)/2$ and solving $(B+\alpha_0I)^\top P_B+P_B(B+\alpha_0I)=-I$. Figure~\ref{fig:E1} shows that the logarithmic-norm bound can fail while a transient-amplification certificate remains valid. Large certificates are visible as conservative bounds, not relabelled as measured errors. All evaluated trajectory errors satisfy the applicable pointwise bounds within numerical tolerance; this numerical check supplements, rather than replaces, their proofs.
\begin{figure}[tb]\centering\includegraphics[width=\textwidth]{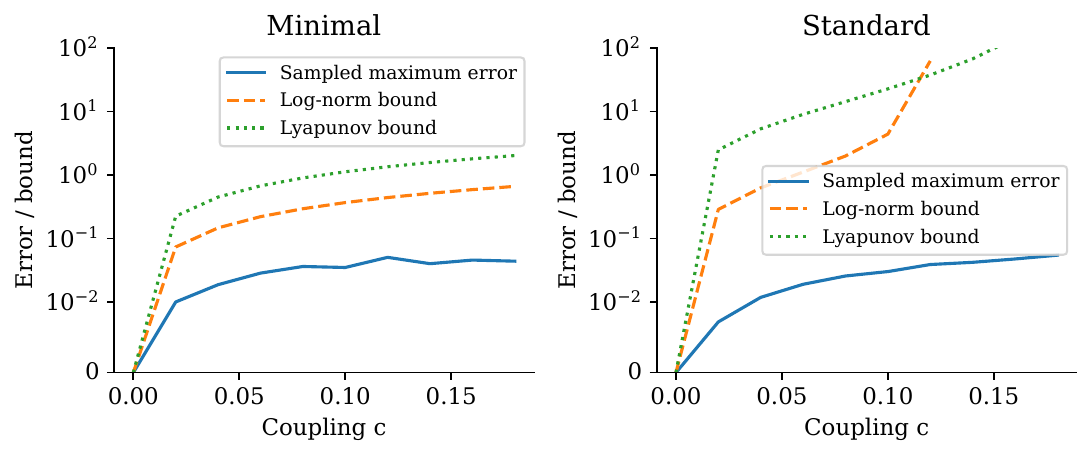}
\caption{E1: sampled maximum errors and two valid asymptotic error bounds. The logarithmic-norm curve is omitted when its contraction hypothesis fails. A retained slow, non-normal subsystem can make a sufficient certificate conservative despite small observed error.}\label{fig:E1}\end{figure}

\subsection{E2: recovery-time ordering and its reversal}
At zero coupling, an initial displacement $0.2(\sqrt{1-\varphi^2}e_M+\varphi(e_R+e_Q)/\sqrt2)$ gives a hidden slow-mode experiment for $\varphi=0,0.25,0.5,0.75,1$. The norm reduces to that of two decaying coordinates with rates $0.4$ and $0.05$. Restoration times are obtained from the continuous-time analytic norm by scalar root finding. Table~\ref{tab:E2} reports separately normalised times. At $\varphi=1$ the minimal observation sees no displacement.
\begin{table}[tb]\centering
\caption{E2: relative restoration times at tolerance 0.1, computed in continuous time. Time units are the synthetic model's days.}\label{tab:E2}
\begin{tabular}{rrr}\toprule Hidden fraction $\varphi$ & Full state & Minimal observation\\\midrule
0.00 & 5.756 & 5.756\\
0.25 & 18.326 & 5.756\\
0.50 & 32.189 & 5.756\\
0.75 & 40.298 & 5.756\\
1.00 & 46.052 & no visible displacement\\\bottomrule
\end{tabular}\end{table}
The diagonal counterexample of Proposition~\ref{prop:recovery} is also computed over tolerances from $10^{-1}$ to $10^{-8}$. At tolerance $0.1$, the reduced relative time is approximately ten times the full time. Their ratio approaches one as the tolerance tends to zero because both trajectories excite the retained slow eigenvalue. In both examples the common-absolute-threshold ordering holds, as Figure~\ref{fig:E2} illustrates.
\begin{figure}[tb]\centering\includegraphics[width=\textwidth]{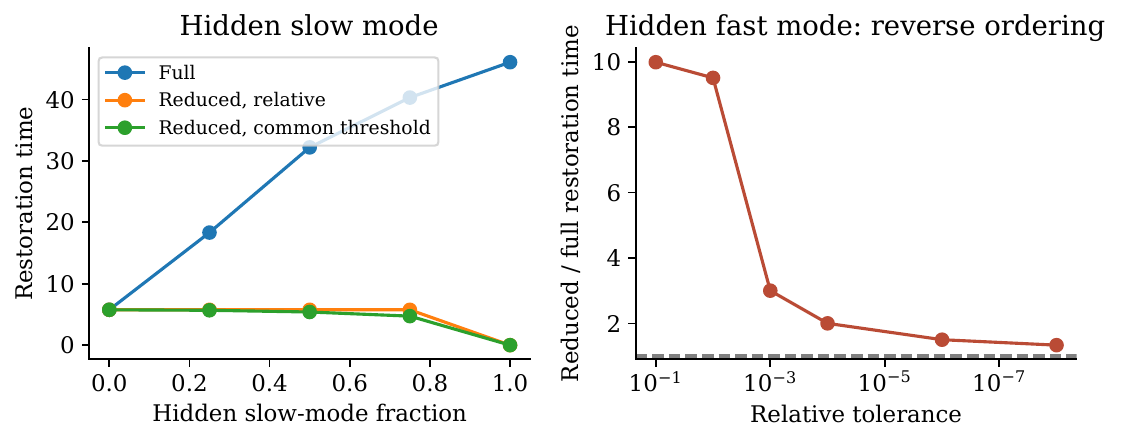}
\caption{E2: omission of a slow mode shortens reported recovery in the first example. Omission of a large-amplitude fast mode reverses the relative-threshold ordering in the second. The latter ratio tends to one at small tolerance, consistently with the asymptotic result.}\label{fig:E2}\end{figure}

\subsection{E3: nonlinear coupled fixed-point sensitivity}
We replace a clamped affine illustration by the smooth nonlinear map $T_q(v)=\tanh(K_qv)$, componentwise, where
\begin{equation}K_q=\begin{pmatrix}0.25&q-0.25\\q-0.25&0.25\end{pmatrix},\qquad 0.4\le q\le0.99.\label{eq:nonlinear_map}\end{equation}
It maps $[-1,1]^2$ into itself and contracts in the maximum norm with gain at most $q$. Its fixed point is zero. We use 100 uniform initialisations at each of seven gains and stop on actual fixed-point error relative to its initial value, with tolerance $10^{-8}$. For the perturbed map $\widetilde T_q(v)=\tanh(K_qv+0.001{\bf1})$, the defect at zero is at most $0.001$, giving displacement bound $0.001/(1-q)$.

Actual iteration counts remain below the contraction bound, and computed displacements remain below the defect bound (Figure~\ref{fig:E3}). At $q=0.99$, saturation makes the displacement substantially smaller than the linear upper bound. This experiment supports the certificate under stated conditions; it does not imply that every system bifurcates when a comparison gain reaches one.
\begin{figure}[tb]\centering\includegraphics[width=\textwidth]{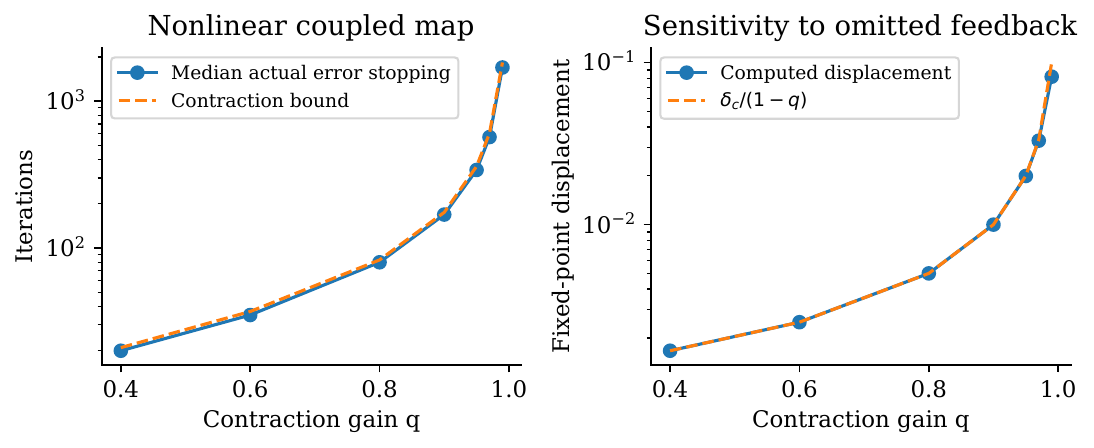}
\caption{E3: nonlinear iteration and equilibrium perturbation. The left panel uses actual fixed-point error, avoiding ambiguity from an increment-based stopping test. The right panel compares displacement against the contraction certificate.}\label{fig:E3}\end{figure}

\subsection{E4 and E5: cadence-aware information and channel selection}
The synthetic effect direction is proportional to $(0.60,0.50,0.40,0.15,0,0,0,0,0,0,0,0.15)$ in the order of \eqref{eq:state}, and normalised to unit Euclidean length. There is one identity-loading channel per construct; this is an explicit synthetic sensing assumption. Sensor standard deviations are $0.02$ for $M,S,D_{\rm eff}$, $0.05$ for $R,C,Q,I$ and $0.10$ for the remaining coordinates. The last group is observed on days $0,7,14,21$; other channels are observed on days $0,1,\ldots,29$.

For times $t_i,t_j$ in one channel, the covariance is
\begin{equation}\Omega_{ij}=0.09^2\,0.6^{|t_i-t_j|}+\sigma^2\mathbf1\{i=j\}.\label{eq:covariance}\end{equation}
Channels are independent of each other in this experiment. Weekly observations retain the appropriate nonzero lag-seven covariance. Formula~\eqref{eq:info} uses exactly this same covariance and cadence as the Monte Carlo generator. There is no approximation that treats all channels as daily, and no claim that the statistical windows were generated by the separate recovery matrix in E1.

At each tier, 20\,000 independent Gaussian windows are generated; their centred noise estimates are reused across the eight effect sizes to obtain paired power curves. Each individual rejection-rate estimate retains its binomial interpretation, but estimates across effects are correlated. The test level is $0.05$. Analytic and empirical results are shown in Figure~\ref{fig:E45}; numerical values are generated directly from the output files in Table~\ref{tab:power}.
\begin{table}[tb]\centering
\caption{E4: window information and Monte Carlo rejection rates. Each rate uses 20\,000 windows; power is evaluated at effect 0.10.}\label{tab:power}
\begin{tabular}{lrrr}\toprule Support & Information & Empirical size & Empirical power\\\midrule
Minimal & 443.231 & 0.0495 & 0.6856 \\
Standard & 780.922 & 0.0492 & 0.8773 \\
Full & 823.494 & 0.0510 & 0.8869 \\
\bottomrule\end{tabular}\end{table}

For E5, the target is 80\% power at effect $0.10$, requiring window information $J_S\ge618.26$ when bias and additional calibration noise are zero. Greedy ranking is checked against all 4095 nonempty supports, under baseline noise, doubled survey noise and halved sensor noise. In each scenario, $\{M,R\}$ is a minimum-cardinality statistical support. The minimal catalogue cannot attain the target because only its $M$ channel loads on this chosen effect direction. These results depend on that direction and covariance; they do not establish a general municipal minimum dataset.

The correlated example \eqref{eq:correlated_example} is enumerated separately. Its optimal budget-two information is approximately 22.5 times the greedy value. This demonstrates why exact optimality in the independent-channel example cannot be extended to arbitrary covariance structures.
\begin{figure}[tb]\centering\includegraphics[width=\textwidth]{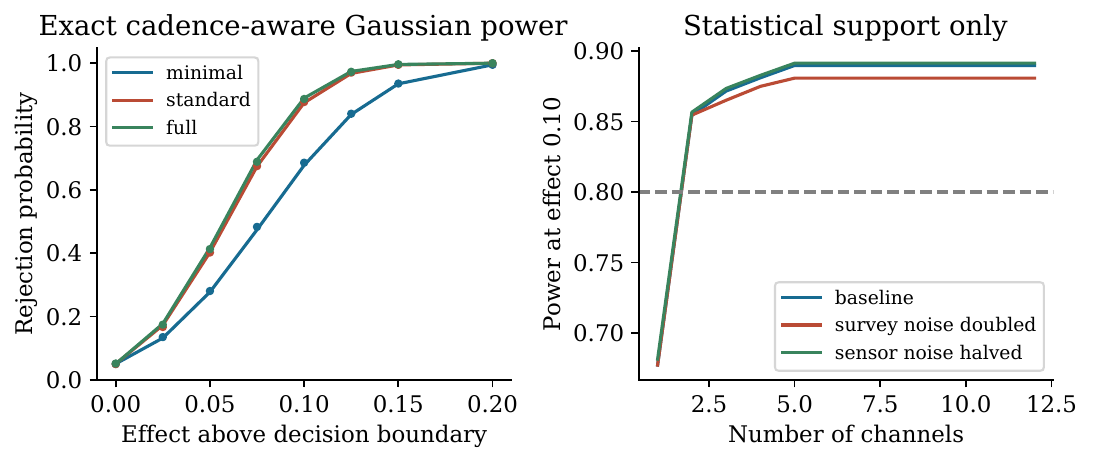}
\caption{E4--E5: exact analytic power and Monte Carlo estimates with pointwise 95\% binomial normal-approximation error bars (left); information-only support selection under three synthetic noise scenarios (right). The smallest statistical support need not satisfy the dynamical certificate.}\label{fig:E45}\end{figure}

\subsection{E6: observation charts and nonlinear reduction comparators}
The first part evaluates \eqref{eq:parabola}. Perturbing $s$ by $10^{-7}$ shows the expected amplification as $a$ approaches $-1/2$: the inverse derivative is 50 at $a=-0.49$ and 1 at $a=0$. This verifies chart conditioning, not a universal relationship between conditioning and closure error.

A separate tube calculation checks the bound in Proposition~\ref{prop:chart}. Use $a\in[-0.45,0.35]$, $\beta=0.1$, $L_v=\sqrt{1.81}$ and $L_f=2.2$, the latter justified by a Frobenius bound for the polynomial Jacobian on the enclosing rectangle. Evaluate $x\in[-0.4,0.3]$ and $y=x^2+w$ on 60\,501 grid points for each $|w|\le r$, with $r=0.001,0.003,0.005$. The certified defect bounds are respectively $0.06231,0.18692,0.31154$, compared with sampled defects $0.00142,0.00443,0.00776$. The bound is conservative; the grid maxima are checks within it, not substitutes for a supremum proof. Complete constants and domain checks are supplied in the code.

For a separate nonlinear comparison with a known exact phase, construct
\begin{equation}\dot u=-0.2u-u^3,\qquad\dot w=-3w,\qquad x=u+0.3w,\quad y=u^2+w.\label{eq:benchmark}\end{equation}
On the chosen branch $1-0.6u>0$, the physical coordinate transform is invertible. The manifold $w=0$ is $y=x^2$, with reduced dynamics $\dot u=-0.2u-u^3$ and exact solution
\begin{equation}u(t)=\frac{u_0e^{-0.2t}}{\sqrt{1+5u_0^2(1-e^{-0.4t})}}.\label{eq:slowexact}\end{equation}
The full observed coordinate is $x(t)=u(t)+0.3w_0e^{-3t}$. Its physical vector field is obtained by differentiating the transform, not by fitting a manifold to the test outputs.

Four comparators are evaluated: the exact manifold flow initialised at the observed $x_0$; the same flow initialised at its correct phase $u_0$; a least-squares cubic closure fitted to observed $x$ and its derivative; and rank-one proper orthogonal decomposition (POD) with Galerkin dynamics in physical $(x,y)$ coordinates. The phase-initialised comparator uses the full initial state. POD uses full training snapshots and the full vector field, but is initialised at test time to match observed $x_0$. The fitted closure uses training $x$ and $\dot x$. Their information privileges are therefore different and are stated explicitly; this is a diagnostic comparison rather than a claim of equal-cost superiority.

The 16 training trajectories use the Cartesian product of
\[
u_0\in\{-0.3,-0.1,0.1,0.3\},\qquad w_0\in\{-0.15,-0.05,0.05,0.15\},
\]
with 161 times over $[0,8]$. The 16 distinct test states use
\[
u_0\in\{-0.25,-0.15,0.15,0.25\},\qquad w_0\in\{-0.12,-0.04,0.04,0.12\},
\]
with 801 evaluation times. Numerical reduced flows use DOP853 with relative tolerance $10^{-10}$ and absolute tolerance $10^{-12}$. Physical full-system integration is checked against \eqref{eq:slowexact}; tightening the POD tolerances by a factor of 100 checks numerical convergence.

\begin{table}[tb]\centering
\caption{E6: mean maximum observed-coordinate error over $2\le t\le8$, across 16 held-out initial states.}\label{tab:methods}
\begin{tabular}{lr}\toprule Method & Mean late maximum error\\\midrule
Manifold, observed initial condition & 0.0136211 \\
Manifold, exact phase initial condition & 0.0000595 \\
Fitted cubic closure & 0.0129519 \\
Rank-one POD--Galerkin & 0.0538988 \\
\bottomrule\end{tabular}\end{table}

The phase-initialised model has observed-coordinate error exactly $0.3|w_0|e^{-3t}$; it sacrifices initial observation matching to obtain the correct long-term trajectory. The other methods can fit the immediate observation while retaining a slow phase error. Table~\ref{tab:methods} and Figure~\ref{fig:E6} quantify this distinction. An instantaneous fixed observation does not contain the hidden variable needed for exact phase recovery. Repeated observations or an estimated observer could alter that limitation, but are outside this benchmark.
\begin{figure}[tb]\centering\includegraphics[width=\textwidth]{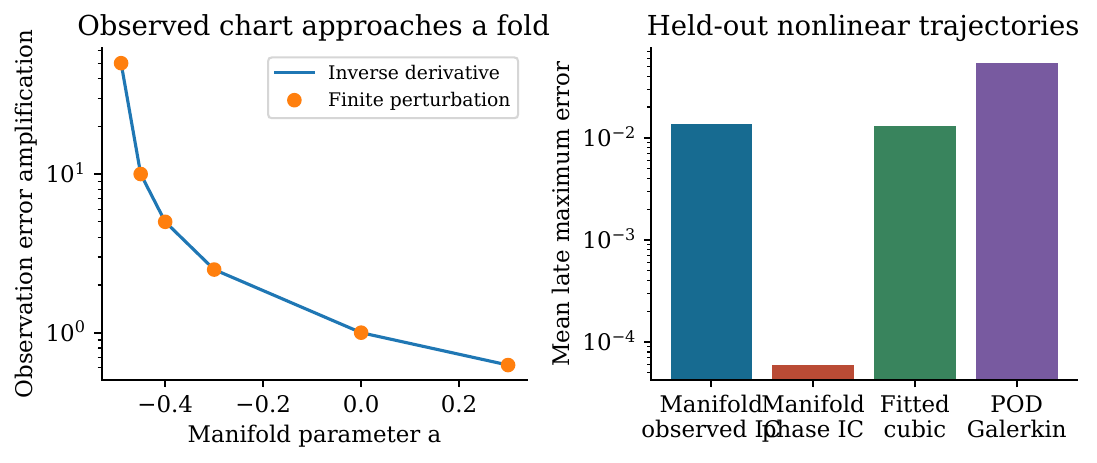}
\caption{E6: conditioning of Roberts's observed parabola chart (left) and reduction errors in the separate constructed nonlinear benchmark (right). Late error is the maximum over $2\le t\le8$, averaged across 16 held-out initial conditions. The phase comparator has additional initial-state information.}\label{fig:E6}\end{figure}

\subsection{E7: the joint certificate in a nonlinear example}
To exercise the combined bound rather than only its separate components, take a recovery coordinate
\begin{equation}\dot x=-x-x^3+c z,\quad\dot z=-(z-0.2),\quad x(0)=0,\ z(0)=0.2,\label{eq:joint_benchmark}\end{equation}
with observed closure $\dot y=-y-y^3$, $y(0)=0$. The interval $|x|\le0.2$ is invariant for $0\le c\le0.1$, the closure contracts at rate at least one, and $z(t)=0.2$. Thus $|x(5)-y(5)|\le0.2c(1-e^{-5})$. Pair this with a scalar strategic map $T_q(v)=qv$, full fixed point zero, and reduced map $\widetilde T_q(v)=qv+0.001$ on $\R$. Its exact displacement is $0.001/(1-q)$.

Choose a margin observable $H(r,v)=\theta_0-r+v$. The full margin is $\theta=\theta_0-x(5)$ and the reduced one is $\tilde\theta=\theta_0+0.001/(1-q)$. Their actual upward bias is $x(5)+0.001/(1-q)$, bounded by
\begin{equation}B=0.2c(1-e^{-5})+\frac{0.001}{1-q}.\label{eq:E7bias}\end{equation}
The intercept $\theta_0$ sets the desired full null or alternative margin without changing the dynamics. Attach an independent scalar Gaussian measurement of $\tilde\theta$ with information $J=800$ and $\sigma_{\rm cal}=0$. This specifies the observation model needed for the test, rather than inferring it from a deterministic trajectory.

For $q\in\{0.5,0.8,0.95\}$ and $c\in\{0,0.02,0.05,0.1\}$, a naive unadjusted test can have false-positive probability up to 0.3024 at a nominal level of 0.05. The bias-adjusted rule remains at or below 0.05 under the stipulated model. We additionally simulate the least favourable bias $-B$ to verify the worst-case power formula over the bounded-bias uncertainty class; this sign is an adversarial statistical scenario, not the bias produced by the displayed positive-forcing dynamics. Each configuration uses 50\,000 Gaussian replications. Figure~\ref{fig:E7} shows how reliable power falls as either omitted feedback or strategic sensitivity increases.
\begin{table}[tb]\centering\small
\caption{E7: selected joint-certificate cases. Robust power is the analytic minimum over $|b|\le B$ at effect 0.10 and information 800; $J_{\rm req}$ is the information required for 80\% worst-case power.}\label{tab:joint}
\begin{tabular}{rrrrrr}\toprule $q$ & $c$ & $B$ & Naive size & Robust power & $J_{\rm req}$\\\midrule
0.50 & 0.00 & 0.0020 & 0.0561 & 0.8578 & 670.9 \\
0.50 & 0.10 & 0.0219 & 0.1523 & 0.4787 & 1952.6 \\
0.80 & 0.00 & 0.0050 & 0.0664 & 0.8161 & 763.3 \\
0.80 & 0.10 & 0.0249 & 0.1732 & 0.4118 & 2446.6 \\
0.95 & 0.00 & 0.0200 & 0.1403 & 0.5208 & 1717.4 \\
0.95 & 0.10 & 0.0399 & 0.3024 & 0.1420 & 15048.1 \\
\bottomrule\end{tabular}\end{table}
\begin{figure}[tb]\centering\includegraphics[width=\textwidth]{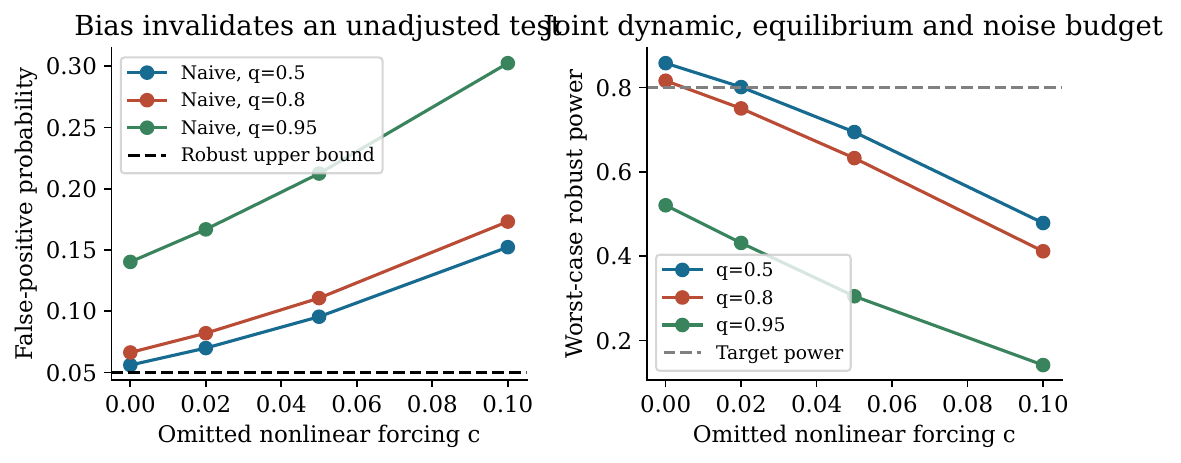}
\caption{E7: false-positive inflation from the actual positive model bias (left) and worst-case power of the robust test over both bias signs (right). The nominal level is 0.05, information is 800, and the alternative full margin is 0.10.}\label{fig:E7}\end{figure}

\section{Discussion}
The results separate four questions that are easily conflated: whether observed dynamics close, whether the reduced equilibrium represents the projected full one, whether the decision observable can be reconstructed, and whether its margin can be estimated with adequate power. A positive answer to one is not a certificate for the others. The combined test makes their relationship explicit through a deterministic bias allowance and a stochastic variance allowance.

This formulation changes how a limited indicator set should be assessed. If the dominant problem is measurement variance, improved cadence or additional informative channels may suffice. If the dominant term is $\eta_h$, hidden decision variables must be observed or bounded through additional assumptions. If the dominant term is $\delta_c/(1-q)$, near-critical strategic sensitivity makes even a small model mismatch consequential. A certificate may also be improved by a less conservative contraction norm or a better closure, without acquiring new data. Failure of a particular certificate should therefore trigger a diagnosis, not an automatic conclusion that the reduction is invalid.

The distinction from manifold methods is about the guarantee and available information. On a known invariant manifold with a valid observed chart, fixed observations can yield exact reduced dynamics. Off the manifold, asymptotic phase initialisation can provide much better long-term predictions than immediate observation matching, as E6 demonstrates. A defect-based guarantee is useful when one must keep the observation coordinate and can bound unresolved forcing; it does not replace manifold construction when that structure is available. The chart-to-defect estimate identifies a sufficient connection without conflating geometry with dynamical error.

The statistical support result is deliberately limited. The two informative channels in E5 are sufficient for a chosen scalar effect, known covariance and zero model bias. They are not sufficient evidence for full-state identifiability, causal antifragility or preservation of a coupled traffic equilibrium. With unknown effect direction, a scalar Fisher information can be misleading: information for multiple parameters, nuisance-adjusted information or a worst-direction design would be needed. Likewise, a known parameter covariance cannot be inferred from qualitative accessibility labels alone.

Five limitations delimit the contribution. First, the theory assumes a model and operating domain; sparse calibration data cannot certify unseen dynamics without additional regularity or uncertainty assumptions. Secondly, the combined error budget is additive and conservative; correlations among model discrepancies, one-sided bias information and sharper goal-oriented bounds may improve it. Thirdly, exact Gaussian power assumes known covariance and independently fixed design; heavy tails, estimated baselines and online support selection require further inference. Fourthly, the nonlinear benchmarks are small systems with known structure. They establish computational behaviour under controlled assumptions, not scalability to arbitrary traffic networks. Fifthly, the analysis covers equilibria and their local recovery, not preservation of general periodic, chaotic or non-autonomous attractors.

These limits also identify a focused research programme. Useful extensions include calibrated high-probability defect bounds, sequential observation policies with controlled testing error, and efficient joint support search when the deterministic bias allowance depends on the selected channels. Empirical work would need a documented indicator assembly, baseline design, observation covariance and external validation of decision outcomes. None of those field claims is supplied by the present synthetic study.

\section{Conclusion}
Partial observation can support an equilibrium decision when its dynamical, strategic, observational and statistical errors are jointly controlled. The combined certificate gives an explicit sufficient bias bound and an exact worst-case power requirement within a stated Gaussian model. It also identifies irreducible ambiguity on observation fibres, distinguishes observation charts from phase projections, and shows why recovery-time comparisons need a declared tolerance convention. Reproducible synthetic experiments expose both successful certificates and failure modes. The resulting framework provides a mathematical basis for judging a particular monitoring design and decision, while keeping that judgement separate from empirical validation and causal claims of antifragility.

\section*{Ethics}
This work is theoretical and computational. It did not involve human participants, personal data or animal subjects, and no ethical approval was required. All data used are synthetic and generated by the accompanying code.

\section*{Data accessibility}
All inputs are synthetic. The accompanying reproducibility package contains the complete executable experiment code, CSV outputs, numerical verification records, generated figures, dependency versions and instructions. The archived version-8 code and outputs are separated from the revised experiments. No empirical city data are used. The supplied electronic supplementary material is intended to accompany the submission and be made publicly available with the article.
\section*{Authors' contributions}
Ali Ghoroghi: conceptualisation, methodology, formal analysis, software, investigation, writing (original draft), writing (review and editing). Yacine Rezgui: supervision, funding acquisition, writing (review and editing). Afrouz Ghaemi: methodology, investigation, validation, writing (review and editing).
\section*{Declaration of AI assistance}
During manuscript preparation, the authors used Claude (Anthropic) and ChatGPT (OpenAI) for language editing and technical assistance, including drafting candidate mathematical statements and proofs, developing computational code, and implementing and executing synthetic numerical experiments. The authors critically reviewed and verified the mathematical arguments, code, numerical results and interpretations, and take full responsibility for the scientific content and conclusions of the article.

\section*{Competing interests}
The authors declare no competing interests.
\section*{Funding}
This work was supported by the European Union's Horizon Europe research and innovation programme under the AntifragiCity project, Grant Agreement No.~101203052. Views and opinions expressed are those of the authors only and do not necessarily reflect those of the European Union or the granting authority.
\section*{Acknowledgements}
The authors thank A.~J.~Roberts for correspondence clarifying invariant-manifold parametrisation and appropriate initial conditions, and project colleagues for the motivating monitoring context. No empirical project data are used as validation datasets.


\begin{thebibliography}{31}
\providecommand{\natexlab}[1]{#1}
\providecommand{\url}[1]{\texttt{#1}}
\expandafter\ifx\csname urlstyle\endcsname\relax
  \providecommand{\doi}[1]{doi: #1}\else
  \providecommand{\doi}{doi: \begingroup \urlstyle{rm}\Url}\fi

\bibitem[Li and Rabitz(1989)]{LiRabitz1989}
Genyuan Li and Herschel Rabitz.
\newblock A general analysis of exact lumping in chemical kinetics.
\newblock \emph{Chemical Engineering Science}, 44\penalty0 (6):\penalty0
  1413--1430, 1989.
\newblock \doi{10.1016/0009-2509(89)85014-6}.

\bibitem[Tomlin et~al.(1997)Tomlin, Li, Rabitz, and T{\'o}th]{Tomlin1997}
Alison~S. Tomlin, Genyuan Li, Herschel Rabitz, and J{\'a}nos T{\'o}th.
\newblock The effect of lumping and expanding on kinetic differential
  equations.
\newblock \emph{SIAM Journal on Applied Mathematics}, 57\penalty0 (6):\penalty0
  1531--1556, 1997.
\newblock \doi{10.1137/S0036139995293294}.

\bibitem[Leguizamon-Robayo et~al.(2025)Leguizamon-Robayo, Jim{\'e}nez-Pastor,
  Tribastone, Tschaikowski, and Vandin]{Leguizamon2025}
Alexander Leguizamon-Robayo, Antonio Jim{\'e}nez-Pastor, M.~Tribastone, Max
  Tschaikowski, and Andrea Vandin.
\newblock Approximate constrained lumping of chemical reaction networks.
\newblock \emph{Proceedings of the Royal Society A}, 481\penalty0
  (2317):\penalty0 20240754, 2025.
\newblock \doi{10.1098/rspa.2024.0754}.

\bibitem[Carr(1981)]{Carr1981}
Jack Carr.
\newblock \emph{Applications of Centre Manifold Theory}, volume~35 of
  \emph{Applied Mathematical Sciences}.
\newblock Springer-Verlag, New York, 1981.

\bibitem[Fenichel(1979)]{Fenichel1979}
Neil Fenichel.
\newblock Geometric singular perturbation theory for ordinary differential
  equations.
\newblock \emph{Journal of Differential Equations}, 31\penalty0 (1):\penalty0
  53--98, 1979.
\newblock \doi{10.1016/0022-0396(79)90152-9}.

\bibitem[Kuehn(2015)]{Kuehn2015}
Christian Kuehn.
\newblock \emph{Multiple Time Scale Dynamics}, volume 191 of \emph{Applied
  Mathematical Sciences}.
\newblock Springer, Cham, 2015.
\newblock \doi{10.1007/978-3-319-12316-5}.

\bibitem[Roberts(2015)]{Roberts2015}
A.~J. Roberts.
\newblock \emph{Model Emergent Dynamics in Complex Systems}.
\newblock Mathematical Modeling and Computation. SIAM, Philadelphia, PA, 2015.
\newblock \doi{10.1137/1.9781611973563}.

\bibitem[Lohmiller and Slotine(1998)]{LohmillerSlotine1998}
Winfried Lohmiller and Jean-Jacques~E. Slotine.
\newblock On contraction analysis for non-linear systems.
\newblock \emph{Automatica}, 34\penalty0 (6):\penalty0 683--696, 1998.
\newblock \doi{10.1016/S0005-1098(98)00019-3}.

\bibitem[S{\"o}derlind(2006)]{Soderlind2006}
Gustaf S{\"o}derlind.
\newblock The logarithmic norm. {History} and modern theory.
\newblock \emph{BIT Numerical Mathematics}, 46\penalty0 (3):\penalty0 631--652,
  2006.
\newblock \doi{10.1007/s10543-006-0069-9}.

\bibitem[Kay(1993)]{Kay1993}
Steven~M. Kay.
\newblock \emph{Fundamentals of Statistical Signal Processing: Estimation
  Theory}.
\newblock Prentice Hall, Upper Saddle River, NJ, 1993.

\bibitem[Pukelsheim(1993)]{Pukelsheim1993}
Friedrich Pukelsheim.
\newblock \emph{Optimal Design of Experiments}.
\newblock John Wiley \& Sons, New York, 1993.

\bibitem[Lehmann and Romano(2005)]{LehmannRomano2005}
Erich~L. Lehmann and Joseph~P. Romano.
\newblock \emph{Testing Statistical Hypotheses}.
\newblock Springer, New York, 3rd edition, 2005.

\bibitem[Armstrong and Koles{\'a}r(2020)]{ArmstrongKolesar2020}
Timothy~B. Armstrong and Michal Koles{\'a}r.
\newblock Simple and honest confidence intervals in nonparametric regression.
\newblock \emph{Quantitative Economics}, 11\penalty0 (1):\penalty0 1--39, 2020.
\newblock \doi{10.3982/QE1199}.

\bibitem[Ghoroghi et~al.(2026)Ghoroghi, Rezgui, and
  Ghaemi]{GhoroghiPreprint2026}
Ali Ghoroghi, Yacine Rezgui, and Afrouz Ghaemi.
\newblock Multi-scale equilibrium under variable indicator dimensionality:
  Faithful reduction of dynamic attractors in urban mobility systems, 2026.
\newblock Preprint.

\bibitem[Rezgui et~al.(2025)Rezgui, Ghoroghi, Manthos, Vlachopanagiotis, and
  Ghaemi]{Rezgui2025}
Yacine Rezgui, Ali Ghoroghi, Evangelos Manthos, Theodoros Vlachopanagiotis, and
  Afrouz Ghaemi.
\newblock {AntifragiCity}: Advancing urban mobility through antifragile
  principles.
\newblock In \emph{2025 IEEE International Conference on Engineering,
  Technology, and Innovation (ICE/ITMC)}, 2025.
\newblock \doi{10.1109/ICE/ITMC65658.2025.11106596}.

\bibitem[Beckmann et~al.(1956)Beckmann, McGuire, and Winsten]{Beckmann1956}
Martin Beckmann, C.~B. McGuire, and Christopher~B. Winsten.
\newblock \emph{Studies in the Economics of Transportation}.
\newblock Yale University Press, New Haven, CT, 1956.

\bibitem[Traub et~al.(1988)Traub, Wasilkowski, and
  Wo{\'z}niakowski]{TraubWasilkowskiWozniakowski1988}
Joseph~F. Traub, Grzegorz~W. Wasilkowski, and Henryk Wo{\'z}niakowski.
\newblock \emph{Information-Based Complexity}.
\newblock Academic Press, Boston, MA, 1988.

\bibitem[Becker and Rannacher(2001)]{BeckerRannacher2001}
Roland Becker and Rolf Rannacher.
\newblock An optimal control approach to a posteriori error estimation in
  finite element methods.
\newblock \emph{Acta Numerica}, 10:\penalty0 1--102, 2001.
\newblock \doi{10.1017/S0962492901000010}.

\bibitem[Roberts(1989)]{Roberts1989}
A.~J. Roberts.
\newblock Appropriate initial conditions for asymptotic descriptions of the
  long term evolution of dynamical systems.
\newblock \emph{Journal of the Australian Mathematical Society. Series B.
  Applied Mathematics}, 31\penalty0 (1):\penalty0 48--75, 1989.
\newblock \doi{10.1017/S0334270000006470}.

\bibitem[Holmes et~al.(1996)Holmes, Lumley, and
  Berkooz]{HolmesLumleyBerkooz1996}
Philip Holmes, John~L. Lumley, and Gahl Berkooz.
\newblock \emph{Turbulence, Coherent Structures, Dynamical Systems and
  Symmetry}.
\newblock Cambridge University Press, Cambridge, 1996.

\bibitem[Antoulas(2005)]{Antoulas2005}
Athanasios~C. Antoulas.
\newblock \emph{Approximation of Large-Scale Dynamical Systems}.
\newblock Advances in Design and Control. SIAM, Philadelphia, PA, 2005.
\newblock \doi{10.1137/1.9780898718713}.

\bibitem[Kivel{\"a} et~al.(2014)Kivel{\"a}, Arenas, Barthelemy, Gleeson,
  Moreno, and Porter]{Kivela2014}
Mikko Kivel{\"a}, Alex Arenas, Marc Barthelemy, James~P. Gleeson, Yamir Moreno,
  and Mason~A. Porter.
\newblock Multilayer networks.
\newblock \emph{Journal of Complex Networks}, 2\penalty0 (3):\penalty0
  203--271, 2014.
\newblock \doi{10.1093/comnet/cnu016}.

\bibitem[De~Domenico et~al.(2015)De~Domenico, Nicosia, Arenas, and
  Latora]{DeDomenico2015}
Manlio De~Domenico, Vincenzo Nicosia, Alex Arenas, and Vito Latora.
\newblock Structural reducibility of multilayer networks.
\newblock \emph{Nature Communications}, 6:\penalty0 6864, 2015.
\newblock \doi{10.1038/ncomms7864}.

\bibitem[Boccaletti et~al.(2014)Boccaletti, Bianconi, Criado, del Genio,
  G{\'o}mez-Garde{\~n}es, Romance, Sendi{\~n}a-Nadal, Wang, and
  Zanin]{Boccaletti2014}
Stefano Boccaletti, Ginestra Bianconi, Regino Criado, Charo~I. del Genio,
  Jes{\'u}s G{\'o}mez-Garde{\~n}es, Miguel Romance, Irene Sendi{\~n}a-Nadal,
  Zhen Wang, and Massimiliano Zanin.
\newblock The structure and dynamics of multilayer networks.
\newblock \emph{Physics Reports}, 544\penalty0 (1):\penalty0 1--122, 2014.
\newblock \doi{10.1016/j.physrep.2014.07.001}.

\bibitem[{OECD} and {European Commission Joint Research
  Centre}(2008)]{OECD2008}
{OECD} and {European Commission Joint Research Centre}.
\newblock \emph{Handbook on Constructing Composite Indicators: Methodology and
  User Guide}.
\newblock OECD Publishing, Paris, 2008.
\newblock \doi{10.1787/9789264043466-en}.

\bibitem[Saltelli et~al.(2008)Saltelli, Ratto, Andres, Campolongo, Cariboni,
  Gatelli, Saisana, and Tarantola]{Saltelli2008}
Andrea Saltelli, Marco Ratto, Terry Andres, Francesca Campolongo, Jessica
  Cariboni, Debora Gatelli, Michaela Saisana, and Stefano Tarantola.
\newblock \emph{Global Sensitivity Analysis: The Primer}.
\newblock John Wiley \& Sons, Chichester, 2008.

\bibitem[Walter and Pronzato(1997)]{WalterPronzato1997}
Eric Walter and Luc Pronzato.
\newblock \emph{Identification of Parametric Models from Experimental Data}.
\newblock Springer, London, 1997.

\bibitem[Raue et~al.(2009)Raue, Kreutz, Maiwald, Bachmann, Schilling,
  Klingm{\"u}ller, and Timmer]{Raue2009}
Andreas Raue, Clemens Kreutz, Thomas Maiwald, Julie Bachmann, Marcel Schilling,
  Ursula Klingm{\"u}ller, and Jens Timmer.
\newblock Structural and practical identifiability analysis of partially
  observed dynamical models by exploiting the profile likelihood.
\newblock \emph{Bioinformatics}, 25\penalty0 (15):\penalty0 1923--1929, 2009.
\newblock \doi{10.1093/bioinformatics/btp358}.

\bibitem[Fedorov(1972)]{Fedorov1972}
Valerii~V. Fedorov.
\newblock \emph{Theory of Optimal Experiments}.
\newblock Academic Press, New York, 1972.

\bibitem[Berman and Plemmons(1994)]{BermanPlemmons1994}
Abraham Berman and Robert~J. Plemmons.
\newblock \emph{Nonnegative Matrices in the Mathematical Sciences}, volume~9 of
  \emph{Classics in Applied Mathematics}.
\newblock SIAM, Philadelphia, PA, 1994.
\newblock \doi{10.1137/1.9781611971262}.

\bibitem[Deimling(1985)]{Deimling1985}
Klaus Deimling.
\newblock \emph{Nonlinear Functional Analysis}.
\newblock Springer, Berlin, 1985.
\newblock \doi{10.1007/978-3-662-00547-7}.

\end{thebibliography}
\end{document}